\documentclass{article}

\usepackage{arxiv}
\usepackage{color}
\usepackage[utf8]{inputenc}
\usepackage[T1]{fontenc}
\usepackage{hyperref}
\usepackage{url}
\usepackage{booktabs}
\usepackage{amsfonts}
\usepackage{amsmath}
\usepackage{amssymb}
\usepackage{amsthm}
\usepackage{nicefrac}
\usepackage{microtype}
\usepackage{graphicx}
\usepackage[round,authoryear]{natbib}
\usepackage{doi}
\graphicspath{ {./images/} }
\usepackage{bm}
\usepackage{multirow}

\newtheorem{theorem}{Theorem}

\newtheorem{assumption}{Assumption}

\newcommand{\E}{\mathbb{E}}
\newcommand{\Var}{\mathrm{Var}}

\newcommand{\fhat}{\hat{f}_{\mathrm{ML}}}
\newcommand{\Si}{\hat{S}_{X_i}}
\newcommand{\Ti}{{T}^s_{X_i}}
\newcommand{\hatTi}{\hat{T}^s_{X_i}}
\newcommand{\Vmod}{\hat{V}_{\mathrm{model}}}

\usepackage{fancyhdr}
\title{Learning Global Sensitivity Indices from Observational Data: A Metamodel-Based Approach}

\author{
	Giulia Vannucci\thanks{Corresponding author
		\texttt{giulia.vannucci@unina.it}.} \\
	Department of Electrical Engineering and Information Technology\\
	Polytechnic and Basic Sciences School\\
	University of Naples Federico II\\
	Via Claudio, 21, Napoli (NA) 80125, Italy
}

\begin{document}
	
	\maketitle
	
\begin{abstract}
Classical variance-based Global Sensitivity Analysis (GSA) assumes that the input--output mechanism can be repeatedly evaluated under a designed sampling scheme, which is infeasible when only a given sample of observations is available. We propose MM--GSA, a metamodel-based approach to GSA from observational data, in which supervised learning approximates the systematic input--output relationship. MM--GSA combines two complementary perspectives on input relevance: a model-agnostic estimator of the first-order Sobol' index, quantifying the contribution of an input to the variability of the systematic response, and a new trigger-based structural index, quantifying how predictive performance depends on the availability of a predictor across alternative predictor subsets. We establish consistency for both estimators and a variable-selection property for the structural index under input independence. Monte Carlo experiments and an NHANES application illustrate their finite-sample behavior and show that the two measures provide complementary information on input relevance.


\end{abstract}

	\section{Introduction}
	\label{sec:intro}

Sensitivity Analysis (SA) studies how the outputs of a system are influenced by its inputs and has become a fundamental component of mathematical modeling, uncertainty quantification and decision support \citep{saltelli2021,razavi2021}. Unlike local sensitivity methods, which examine the response to input variations around specific values, Global Sensitivity Analysis (GSA) evaluates the contribution of the inputs over their joint distribution, thereby providing a global characterization of their relevance to the model output \citep{saltelli2008,borgonovo2016common}. Among the many GSA methodologies, variance-based approaches play a central role because they provide an interpretable decomposition of output uncertainty into contributions attributable to individual inputs and their interactions. This decomposition provides a principled basis for ranking input importance, identifying non-influential variables, and guiding model simplification.

Classical variance-based GSA is typically formulated in a setting in which the analyst can repeatedly evaluate the input--output model and control how the input space is explored \citep{saltelli2008,saltelli2010a}. Sensitivity indices are then estimated from model evaluations generated according to structured sampling schemes, in which inputs are deliberately varied to isolate their individual and interaction contributions to output variability. This ability to query the model at chosen input configurations is central to the estimation of classical variance-based sensitivity indices.

In many applications, however, the analyst does not have access to the input--output mechanism itself, but only to a finite dataset of input--output pairs collected in an observational study or obtained from previous model evaluations. The distinction between experimental and observational studies is fundamental in statistical study design \citep{cox2016}. Clinical registries, environmental monitoring, economics, genomics, and machine learning pipelines provide typical examples in which the data have already been collected and the underlying mechanism can no longer be queried at input configurations chosen by the analyst.

From a statistical perspective, such given-data settings naturally lead to supervised learning, with the primary goal of learning the systematic relationship between the inputs and the response. GSA addresses a different question: once this relationship has been learned, how can the contribution of each input to the behavior of the system be quantified? In particular, one may be interested not only in how much an input contributes to the variability of the systematic response, but also in how strongly predictive performance depends on the availability of that input across alternative model specifications. These questions complement statistical learning by using this learned relationship as the basis for a global assessment of input relevance. The methodological challenge is therefore to translate these global sensitivity questions to the observational setting, where the underlying mechanism must first be learned from the available observations and cannot be directly evaluated under a designed sampling scheme.

Several methodological developments are particularly relevant to the present setting. One line of research investigates the relationship between machine learning variable importance and variance-based sensitivity analysis. Earlier work showed that predictive variable importance does not necessarily recover the underlying data-generating mechanism \citep{gottard2020}. More recently, \citet{antoniadis2021} established a theoretical connection between Random Forest permutation importance and the total-order Sobol' index, while trigger-based importance measures have been proposed to distinguish predictive from generative variable importance within the same learning framework \citep{vannucci2026sma}. These contributions establish important connections between sensitivity analysis and variable importance in supervised learning, but their theoretical constructions are intrinsically tied to Random Forests. A second methodological direction extends variance-based sensitivity analysis to dependent inputs through generalized functional ANOVA decompositions \citep{hooker2007,ferrere2026}. These approaches preserve variance decompositions under dependence by introducing orthogonal projections with respect to a reference measure. The resulting sensitivity components, however, are defined through projection-based representations, which can make their interpretation in terms of the original observed predictors less immediate. A particularly relevant connection is the given-data framework of \citet{borgonovo2016common}. Their common rationale shows that a broad class of probabilistic sensitivity measures can be expressed through discrepancies between unconditional and conditional output distributions, and develops consistent estimators from a single sample of input--output observations without requiring a model-in-the-loop. This establishes that global sensitivity questions can be addressed in a given-data setting. Their estimation strategy relies on increasingly refined partitions of the observed input space; the approach developed here follows a different route, using supervised learning metamodels to reconstruct the input--output relationship and evaluate sensitivity through it.

In this paper, we propose a metamodel-based approach to GSA from observational data, which we call MM--GSA. Supervised learning is used to approximate the unknown systematic response function, making it possible to address global sensitivity questions when the underlying input--output mechanism cannot be directly evaluated. The approach combines an established notion of variance-based sensitivity with a new notion of structural sensitivity. Let $X = (X_1, \ldots, X_p)$ denote the vector of input variables, and let $X_i$ denote a generic input of interest. The first component provides $\Si$, a model-agnostic estimator of the first-order sensitivity index, using metamodel predictions to estimate the main-effect function. Examples of metamodel-based GSA can be found in \citet{oakley2004probabilistic} and \citet{marrel2009}, where the metamodel is trained on model evaluations generated under a designed sampling scheme. Here, by contrast, the metamodel is fit directly to a fixed observational sample, and we establish the consistency of the resulting estimator under mild regularity conditions. The second component introduces a new structural sensitivity index, estimated by $\hatTi$, defined through the inclusion indicator associated with $X_i$ across alternative predictor subsets. The superscript $^s$ emphasizes that, unlike the first-order index, this index concerns the presence or absence of $X_i$ in the predictive specification rather than variation in the values of $X_i$ themselves.

The two indices answer distinct but complementary questions: the first-order index quantifies how much of the variability of the systematic response is attributable to $X_i$, whereas the structural index quantifies how strongly predictive performance depends on the availability of $X_i$ across competing model specifications. Their joint use distinguishes input contribution to response variability from the predictive relevance of having that input available. Both indices are defined directly in terms of the observed predictors and can be estimated using different supervised learning algorithms. This preserves their interpretation at the level of the original variables while avoiding projection-based representations or algorithm-specific importance measures.

The remainder of the paper is organized as follows. Section~\ref{sec:s_i_estimator} develops the metamodel-based estimator of the first-order sensitivity index and establishes its theoretical properties. Section~\ref{sec:trigger} introduces the structural sensitivity index $\Ti$, together with its estimator, and studies its statistical properties. Section~\ref{sec:simulations} illustrates the proposed methodology through Monte Carlo studies, and Section~\ref{sec:application} through an illustrative case study on NHANES data. Section~\ref{sec:conclusion} concludes with a discussion of the main findings and future research directions.

\section{Variance-Based Global Sensitivity Analysis}
\label{sec:background}

Variance-based GSA quantifies how uncertainty in the output of a model can be attributed to uncertainty in its input variables. Let $X=(X_1,\ldots,X_p)\in\mathbb R^p$
be a vector of mutually independent input variables, $Y \in \mathbb R$, and let $Y=f(X)$ denote the corresponding model output, where $f:\mathbb R^p\rightarrow\mathbb R$ is the underlying response function. Under this assumption, the functional ANOVA decomposition expresses the output variance as the sum of contributions of increasing dimensionality:
\begin{equation}
V(Y)=\sum_{i=1}^{p}V_i+\sum_{i<j}V_{ij}+\cdots+V_{1,\ldots,p},
\label{eq:decomposition}
\end{equation}
where each term represents the contribution of a subset of inputs to the total output variance \citep{Saltelli1995}.

The corresponding Sobol' sensitivity indices \citep{sobol1995} are obtained by normalizing these variance components:
\begin{equation}
S_i=\frac{V_i}{V(Y)}, \qquad
S_{ij}=\frac{V_{ij}}{V(Y)}, \qquad \ldots
\end{equation}

Among these quantities, the first-order and total-order indices play a central role. The first-order Sobol' index measures the contribution of an input acting alone:
\begin{equation}
S_i=
\frac{\operatorname{Var}_{X_i}
\!\left(
\mathbb E[Y\mid X_i]
\right)}
{\operatorname{Var}(Y)}.
\label{eq:Si_classic}
\end{equation}
The numerator of $S_i$ is the variance of the conditional expectation
$\mathbb E[Y\mid X_i]$, often referred to as the main-effect function. It quantifies how much the expected response changes as $X_i$ varies while the remaining inputs are averaged out. Normalizing by the total output variance yields the fraction of variability that can be attributed to the individual contribution of $X_i$. Equivalently, $S_i$ can be interpreted as the expected reduction in output variance that would be achieved if the value of $X_i$ could be fixed.

The total-order index \citep{homma1996} complements this information by accounting for all interaction effects involving $X_i$:
\begin{equation}
T_i
=
1-
\frac{
\operatorname{Var}_{X_{\sim i}}
\!\left(
\mathbb E[Y\mid X_{\sim i}]
\right)}
{\operatorname{Var}(Y)}
=
\frac{
\mathbb E_{X_{\sim i}}
\!\left(
\operatorname{Var}[Y\mid X_{\sim i}]
\right)}
{\operatorname{Var}(Y)}.
\label{eq:Ti_classic}
\end{equation}
The numerator of the second expression is the expected conditional variance of $Y$ given all inputs except $X_i$. It represents the variability that cannot be explained unless the value of $X_i$ is also specified, and therefore includes both the direct contribution of $X_i$ and all interaction effects involving it. Normalizing by the total output variance yields the overall contribution of $X_i$ to the variability of the model output. Equivalently, $T_i$ can be interpreted as the fraction of output variance that would remain if every factor except $X_i$ could be fixed.

In practice, first-order and total-order indices are estimated through repeated evaluations of the model under structured sampling designs. One of the most widely used estimators is that of \citet{saltelli2010a}:

\begin{equation}
S_i=\frac{V(Y) - \frac{1}{2N} \sum_{v=1}^{N} \left [f (\bm{B})_v - f(\bm{A}_B^{(i)})_v
 \right ]^2}{V(Y)}\,,
 \label{eq:jansen_si}
\end{equation}

\begin{equation}
T_i= \frac{\frac{1}{2N}\sum_{v=1}^{N} \left [ f(\bm{A})_v - f(\bm{A}_B^{(i)})_v \right ]^2}{V(Y)}\,.
 \label{eq:jansen_ti}
\end{equation}
where $f$ is the function or model whose sensitivity is being estimated, $\bm{A}$ and $\bm{B}$ are two independent $N \times p$ sampling matrices with generic elements $a_{vi}$ and $b_{vi}$, where $v=1,\ldots,N$ indexes the rows and $i=1,\ldots,p$ the input variables, and the matrix $\bm{A}_B^{(i)}$ is such that column $i$ comes from matrix $\bm{B}$ and all other $p-1$ columns come from matrix $\bm{A}$.

The measure $S_i$ is computed from the $p$ matrices $\bm{A}_B^{(i)}$ and matrix $\bm{B}$ while measure $T_i$ is computed from the same $p$ matrices $\bm{A}_B^{(i)}$ and matrix $\bm{A}$. In total $N(p+2)$ model evaluations are needed by this design to compute a full set of $S_i$, $T_i$. Further discussion on estimators can be found in~\citet{Owen_2013, azzini2021, kucherenko2011identification}.


\section{Observational Global Sensitivity Analysis}
\label{sec:ob_gsa}

The classical estimators presented in Section~\ref{sec:background} require repeated evaluations of the underlying model under a designed sampling scheme. When only an observational dataset is available, this experimental design is no longer accessible and the classical estimators cannot be applied directly.

Our objective is to estimate variance-based and structural sensitivity measures from observational data. The proposed framework MM--GSA consists of two components. The first develops a metamodel-based estimator of the first-order sensitivity index by approximating the main-effect function through supervised learning metamodels. The second introduces a structural sensitivity index based on trigger functions, which quantifies the contribution of an input through its effect on predictive performance across competing models.

\subsection{First-Order Index via Metamodel Approximation}
	\label{sec:s_i_estimator}

Suppose we wish to estimate the first-order Sobol' index when only observational data are available. The classical definition quantifies the variability of the conditional expectation of the response as one input varies while the remaining inputs are averaged out. In the observational setting, the response function is unknown and cannot be evaluated at arbitrary input configurations. We therefore approximate it using a statistically consistent supervised learning metamodel and estimate the corresponding main-effect function by averaging the metamodel predictions over the empirical distribution of the remaining predictors.

We assume that the observations are generated according to the model
\begin{equation}
Y=f(X)+\varepsilon,
\end{equation}
where $f$ is the unknown response function and $\varepsilon$ is a random error term, independent of $X$, with zero mean.
Let $P_{X,Y}$ denote the joint distribution of $(X,Y)$ and let $\mathcal{D}_{\mathrm{train}}
=\{(\mathbf{x}_j,y_j)\}_{j=1}^{n_{\mathrm{tr}}}$ and $\mathcal{D}_{\mathrm{test}}
=\{(\mathbf{x}_j,y_j)\}_{j=1}^{n_{\mathrm{te}}}$ be two independent samples drawn i.i.d. from $P_{X,Y}$. A supervised learning algorithm trained on $\mathcal{D}_{\mathrm{train}}$ provides a metamodel fit $\hat f:\mathbb R^p\rightarrow\mathbb R$,
which is subsequently used to estimate sensitivity measures from the observations in $\mathcal{D}_{\mathrm{test}}$.
For a given input $X_i$, we denote by $X_{\sim i}=(X_1,\ldots,X_{i-1},X_{i+1},\ldots,X_p)$ 
the vector of all remaining inputs. We write $\mathcal X_i$ and $\mathcal X_{\sim i}$ for the supports of $X_i$ and $X_{\sim i}$, respectively.

The total variance of the metamodel predictions on the test set is estimated as
\begin{equation}
	\Vmod = \frac{1}{n_{te}-1} \sum_{j=1}^{n_{te}} \left(\fhat(\mathbf{x}_j) - \bar{f}\right)^2,
	\label{eq:Vmodel}
\end{equation}
where $\bar{f} = \frac{1}{n_{te}}\sum_{j=1}^{n_{te}} \fhat(\mathbf{x}_j)$. 

The first-order Sobol' index is based on the main-effect function:
\begin{equation}
	f_i(x_i) = \E\left[f(X) \mid X_i = x_i\right]
	= \int_{\mathcal{X}_{\sim i}} 
	f(x_i, \mathbf{x}_{\sim i}) \, 
	dP_{X_{\sim i}|X_i=x_i}(\mathbf{x}_{\sim i})
	\quad x_i \in \mathcal{X}_i.
	\label{eq:fi_true}
\end{equation}
Under input independence, $P_{X_{\sim i}|X_i=x_i} = P_{X_{\sim i}}$, and the conditional expectation reduces to a marginal expectation:
\begin{equation}
	f_i(x_i) = \int_{\mathcal{X}_{\sim i}} 
	f(x_i, \mathbf{x}_{\sim i}) \, 
	dP_{X_{\sim i}}(\mathbf{x}_{\sim i})
	= \E\left[f(x_i, X_{\sim i})\right]
	\label{eq:fi_indep}
\end{equation}
In the following, we assume input independence, under which equations~\eqref{eq:fi_true} and~\eqref{eq:fi_indep} coincide.
The proposed estimator replaces the expectation in equation~\eqref{eq:fi_indep} by an empirical average computed from the metamodel predictions.

To estimate the first-order index, we build a grid of $K \geq 2$ empirical quantiles of $X_i$ on the test set:
\begin{equation}
	x_i^{(1)} \leq x_i^{(2)} \leq \cdots \leq x_i^{(K)}.
	\label{eq:grid}
\end{equation}
For each point $x_i^{(k)}$ on the grid, we fix $X_i = x_i^{(k)}$ and average the metamodel predictions over all test observations for $X_{\sim i}$, where $\mathbf{x}_{j,\sim i} = (x_{j,1},\ldots,x_{j,i-1},x_{j,i+1},\ldots,x_{j,p})$ denotes the $j$-th test observation with the $i$-th component removed. This yields the step function:
\begin{equation}
	\hat{f}_i(x_i^{(k)}) = \frac{1}{n_{te}} 
	\sum_{j=1}^{n_{te}}
	\fhat\!\left(x_i^{(k)},\, \mathbf{x}_{j,\sim i}\right)
	\label{eq:step_Si}
\end{equation}
which approximates the true main effect function $f_i$ defined in equation~\eqref{eq:fi_indep}. 
Finally, the first-order estimator is defined as the variance of the step function normalized by the total variance of the metamodel defined in equation~\eqref{eq:Vmodel}:
\begin{equation}
	\hat{S}_{X_i} = \frac{\displaystyle\frac{1}{K-1}
		\sum_{k=1}^K \left(\hat{f}_i(x_i^{(k)}) - 
		\bar{f}_i\right)^2}{\Vmod}
	\label{eq:Sihat}
\end{equation}
where $\bar{f}_i = \frac{1}{K}\sum_{k=1}^K \hat{f}_i(x_i^{(k)})$ is the mean of the step function on the grid. 

The estimator $\hat{S}_{X_i}$ targets the population quantity
\begin{equation}
S_i^f := \frac{\Var(f_i)}{\Var(f(X))},
\label{eq:Sif}
\end{equation}
the variance of the main-effect function normalized by the variance of the systematic component of the response. By the law of total variance, $\Var(f_i) \leq \Var(f(X))$, so $S_i^f \in [0,1]$. When $\Var(\varepsilon)=0$, $\Var(f(X))=\Var(Y)$ and $S_i^f$ coincides with the classical first-order Sobol' index $S_i=\Var(f_i)/\Var(Y)$; more generally, $S_i^f$ remains well defined and interpretable regardless of the noise level.

\paragraph{Consistency}	
Consistency of $\hat{S}_{X_i}$ with respect to $S_i^f$ follows from decomposing the estimation error into two components:  an approximation error due to replacing $f$ with $\hat{f}_{ML}$, controlled by a consistency assumption on the metamodel, and a Monte Carlo error due to empirical averaging over $X_{\sim i}$, controlled by a uniform law of large numbers for bounded function classes. We work under the following assumptions.

\begin{assumption}[Metamodel consistency]
	\label{ass:H1}
	The ML metamodel $\fhat$ trained on $\mathcal{D}_{train}$ satisfies
	\begin{equation}
		\E\left[(\fhat(X) - f(X))^2\right] \xrightarrow{p} 0 \quad \text{as } n_{tr} \to \infty.
	\end{equation}
\end{assumption}
	
Assumption~\ref{ass:H1} requires that the mean squared error of the metamodel converges to zero as the training sample grows. This condition is satisfied, for example, by Random Forests applied to additive models \citep{scornet2015}, and is assumed here as a general requirement on the metamodel class.
	
\begin{assumption}[Finite fourth moment]
	\label{ass:H2}
	The response function satisfies $\E[f(X)^4] < \infty$.
\end{assumption}
	
Assumption~\ref{ass:H2} is a standard regularity condition ensuring that the sample variance of the metamodel predictions converges to its population counterpart.

\begin{assumption}[Uniform boundedness]
\label{ass:H3}
The true function $f$ and the ML metamodel $\hat{f}_{ML}$ 
satisfy $\|f\|_\infty \leq M$ and 
$\|\hat{f}_{ML}\|_\infty \leq M$ almost surely, 
for some constant $M < \infty$, that is:
\begin{equation}
|f(x_i, \mathbf{x}_{\sim i})| \leq M 
\quad \text{and} \quad
|\hat{f}_{ML}(x_i, \mathbf{x}_{\sim i})| \leq M 
\quad \text{a.s.}
\quad \forall x_i \in \mathcal{X}_i, \; 
\forall \mathbf{x}_{\sim i} \in \mathcal{X}_{\sim i}.
\end{equation}
\end{assumption}

Under Assumption~\ref{ass:H3}, both function classes
\begin{equation}
\mathcal{F}_i = \left\{\mathbf{x}_{\sim i} \mapsto f(x_i, \mathbf{x}_{\sim i}) : x_i \in \mathcal{X}_i \right\} \quad \text{and} \quad \hat{\mathcal{F}}_i = \left\{\mathbf{x}_{\sim i} \mapsto \hat{f}_{ML}(x_i, \mathbf{x}_{\sim i}) : x_i \in \mathcal{X}_i \right\}
\end{equation}
are uniformly bounded by $M$. Their bracketing numbers therefore satisfy $N_{[]}(\varepsilon, \mathcal{F}_i, L_1(P_{X_{\sim i}})) \leq 2M/\varepsilon$ and analogously for $\hat{\mathcal{F}}_i$, for every $\varepsilon > 0$. By Theorem~2.4.1 of \cite{vanderVaart1996}, both $\mathcal{F}_i$ and $\hat{\mathcal{F}}_i$ are Glivenko--Cantelli classes, which yields:
\begin{align}
\sup_{x_i \in \mathcal{X}_i} \left|\frac{1}{n_{te}}\sum_{j=1}^{n_{te}} f(x_i, \mathbf{x}_{j,\sim i}) - \E\left[f(x_i, X_{\sim i})\right]\right| &\xrightarrow{a.s.} 0 \label{eq:uLLN} \\
\sup_{x_i \in \mathcal{X}_i} \left|\frac{1}{n_{te}}\sum_{j=1}^{n_{te}} \hat{f}_{ML}(x_i, \mathbf{x}_{j,\sim i}) - \E\left[\hat{f}_{ML}(x_i, X_{\sim i})\right]\right| &\xrightarrow{a.s.} 0 \label{eq:uLLN_hat}
\end{align}
as $n_{te} \to \infty$.

\begin{assumption}[Deterministic grid]
	\label{ass:H4}
The grid $\{x_i^{(k)}\}_{k=1}^K$ is a deterministic sequence dense in the support $\mathcal{X}_i$ of $X_i$, with $K \to \infty$. 
\end{assumption}
In the numerical implementation the grid consists of the empirical quantiles of $X_i$ on $\mathcal{D}_{test}$. Since empirical quantiles converge to their population counterparts as $n_{te} \to \infty$ \citep[Lemma~21.1]{vanderVaart1998}, a grid of $K$ empirical quantiles becomes dense in $\mathcal{X}_i$ as both $n_{te}$ and $K$ grow, so that Assumption~\ref{ass:H4} holds asymptotically.

\begin{theorem}
	\label{thm:Si}
Under Assumptions~\ref{ass:H1}--\ref{ass:H4}, $\hat{S}_{X_i} \xrightarrow{p} S_i^f$ as $n_{tr}\to\infty$, $n_{te}\to\infty$, and $K\to\infty$.
\end{theorem}

\begin{proof}
From equation~\eqref{eq:Sihat}, we can rewrite the numerator as $\widehat{\Var}(\hat{f}_i)$, the sample variance of the step function $\hat{f}_i(x_i^{(1)}), \ldots, \hat{f}_i(x_i^{(K)})$ over the quantile grid, that is:
\begin{equation*}
    \widehat{\Var}(\hat{f}_i) = 
    \frac{1}{K-1}\sum_{k=1}^K 
    \left(\hat{f}_i(x_i^{(k)}) - \bar{f}_i\right)^2.
\end{equation*}
We can therefore write:
\begin{equation}
    \hat{S}_{X_i} - S_i^f = 
    \frac{\widehat{\Var}(\hat{f}_i)}{\Vmod} - 
    \frac{\Var(f_i)}{\Var(f(X))}
\end{equation}

By adding and subtracting $\Var(f_i) / \Vmod$ we obtain:
	\begin{align}
		\hat{S}_{X_i} - S_i^f 
		&= \frac{\widehat{\Var}(\hat{f}_i) - \Var(f_i)}{\Vmod} 
		+ \Var(f_i) \left(\frac{1}{\Vmod} - \frac{1}{\Var(f(X))}\right) 
		\notag \\
		&= \frac{\widehat{\Var}(\hat{f}_i) - \Var(f_i)}{\Vmod} 
		+ \Var(f_i) \cdot 
		\frac{\Var(f(X)) - \Vmod}{\Vmod \cdot \Var(f(X))}
		\label{eq:decomp}
	\end{align}
Taking absolute values and applying the triangle inequality:
	\begin{equation}
		|\hat{S}_{X_i} - S_i^f| \leq 
		\underbrace{
			\frac{|\widehat{\Var}(\hat{f}_i) - \Var(f_i)|}{\Vmod}
		}_{\mathrm{(A)}}
		+ 
		\underbrace{
			S_i^f \cdot \frac{|\Vmod - \Var(f(X))|}{\Vmod}
		}_{\mathrm{(B)}}
		\label{eq:bound}
	\end{equation}
	
Now we show that both terms converge to zero in probability.
	
Starting with the term $\mathrm{(B)}$ we rewrite:
	\begin{equation}\notag
		S_i^f \cdot \frac{|\Vmod - \Var(f(X))|}{\Vmod} 
		= S_i^f \cdot \left|1 - \frac{\Var(f(X))}{\Vmod}\right|
	\end{equation}
Since $S_i^f \in [0,1]$ is fixed, it suffices to show that $\left|1 - \frac{\Var(f(X))}{\Vmod}\right| \xrightarrow{p} 0$. Since $\Vmod$ is a function of the metamodel predictions $\hat{f}(\mathbf{x}_j)$ alone, its convergence does not involve $Y$ or $\varepsilon$: under Assumptions~\ref{ass:H1} and~\ref{ass:H2}, $\Vmod \xrightarrow{p} \Var(f(X))$ by the Law of Large Numbers applied to the metamodel predictions. Since $\Var(f(X)) > 0$, $\frac{\Var(f(X))}{\Vmod} \xrightarrow{p} \frac{\Var(f(X))}{\Var(f(X))} = 1$.
Therefore, $\left|1 - \frac{\Var(f(X))}{\Vmod}\right| \xrightarrow{p} 0$, and term $\mathrm{(B)} \xrightarrow{p} 0$.

The term $\mathrm{(A)}$ goes to zero when the numerator goes to zero, since the denominator $\Vmod$ is positive by construction and converges in probability to 
$\Var(f(X)) > 0$. We add and subtract $\Var(\hat{f}_i)$, the true variance of the step function of the metamodel, and apply the triangle inequality:
\begin{equation}
	|\widehat{\Var}(\hat{f}_i) - \Var(f_i)| \leq \underbrace{|\widehat{\Var}(\hat{f}_i) - \Var(\hat{f}_i)|}_{\mathrm{(A1)}} + \underbrace{|\Var(\hat{f}_i) - \Var(f_i)|}_{\mathrm{(A2)}}
\end{equation}
Term $\mathrm{(A1)}$ goes to zero by the Law of Large Numbers under Assumption~\ref{ass:H2}.
Term $\mathrm{(A2)}$ goes to zero as follows. 
Under Assumption~\ref{ass:H3}, the uniform convergence 
in equation~\eqref{eq:uLLN_hat} yields 
$\hat{f}_i(x_i^{(k)}) \xrightarrow{a.s.} \hat{f}_i^*(x_i^{(k)})$ 
uniformly over the grid, where 
$\hat{f}_i^*(x_i^{(k)}) = \mathbb{E}[\hat{f}_{ML}(x_i^{(k)}, X_{\sim i})]$.
Under Assumption~\ref{ass:H1}, 
$\hat{f}_i^*(x_i^{(k)}) \xrightarrow{p} f_i(x_i^{(k)})$ 
for each $k$. Together these imply 
$\Var(\hat{f}_i) \xrightarrow{p} \Var(f_i)$.

Since both terms $\mathrm{(A)}$ and $\mathrm{(B)}$ 
converge to zero in probability, we conclude 
$\hat{S}_{X_i} \xrightarrow{p} S_i^f$.
\end{proof}

\paragraph{Relationship with existing methods}

The estimated main-effect function $\hat{f}_i$ defined in equation~\eqref{eq:step_Si} is formally identical to the partial dependence function introduced by \citet{friedman2001}: both are obtained by averaging metamodel predictions over the empirical distribution of $X_{\sim i}$ while fixing $X_i$. In the proposed framework, however, the variance of this function over the quantile grid is used to estimate $\mathrm{Var}(f_i(X_i))$, thereby transforming a descriptive visualization into a quantitative sensitivity measure. In this sense, a partial dependence plot describes \emph{how} $X_i$ influences the response, whereas the proposed estimator quantifies \emph{how much} it contributes to the output variance.

Moreover, the proposed estimator is related to the general framework of given-data sensitivity estimation developed by \citet{borgonovo2016common}, who establish consistency of sensitivity measures of the form $h(F_Y,F_{Y|X_i})$ through partition-refining strategies. Our construction follows a different route: instead of conditioning on partitions that become increasingly finer as the sample size grows, it estimates the main-effect function directly by averaging metamodel predictions over the empirical distribution of the remaining inputs. The two approaches therefore establish consistency through fundamentally different statistical constructions: partition refinement in \citet{borgonovo2016common}, and metamodel approximation in the present work.

\subsection{Trigger-Based Structural Index}
\label{sec:trigger}

Suppose we want to assess the importance of $X_i$ through a question different from the one addressed by $\Si$: how much does predictive performance change when $X_i$ is available or unavailable to the learning algorithm? We propose $\Ti$, a structural sensitivity index that addresses this question through random perturbations of predictor availability. $\Ti$ extends the GSA-based variable selection framework of \citet{becker2021varsel} to the setting of ML metamodels, treating the fitted metamodel as a black box whose predictive loss is evaluated across random input subsets.

Let $\bm{\gamma} = (\gamma_1, \ldots, \gamma_p) \in \{0,1\}^p$ be a binary vector indicating which input variables are included in the model, with $\gamma_j = 1$ meaning $X_j$ is included and $\gamma_j = 0$ meaning it is excluded. For a given $\bm{\gamma}$, let $q(\bm{\gamma})$ denote the predictive loss of a supervised learning algorithm trained on the variables selected by $\bm{\gamma}$, evaluated on observations not used for training. In practice, $q(\bm{\gamma})$ is computed on a held-out 
validation set:
\begin{equation}
q(\bm{\gamma}) = \sqrt{\frac{1}{n_{te}} \sum_{j=1}^{n_{te}} \left(y_j - \hat{y}_j(\bm{\gamma})\right)^2},
\label{eq:q_gamma}
\end{equation}
where $\hat{y}_j(\bm{\gamma})$ denotes the prediction of $y_j$ from the metamodel trained on the variables selected by $\bm{\gamma}$, evaluated on held-out observations. Let $q^*(\bm{\gamma})$ denote the population counterpart of $q(\bm{\gamma})$, that is, the predictive loss achieved by the best predictor of $Y$ based on the variables selected by $\bm{\gamma}$. For a given index $i$, let $\bm{\gamma}^{(i)}$ denote the vector obtained from $\bm{\gamma}$ by flipping the $i$-th coordinate, that is $\gamma^{(i)}_j = \gamma_j$ for $j \neq i$ and $\gamma^{(i)}_i = 1 - \gamma_i$. The trigger index is then estimated as:
\begin{equation}
\hatTi = \frac{\displaystyle\frac{1}{4L}\sum_{l=1}^{L}\left(q(\bm{\gamma}_l^{(i)}) - q(\bm{\gamma}_l)\right)^2}{\displaystyle\frac{1}{L-1}\sum_{l=1}^{L}\left(q(\bm{\gamma}_l) - \bar{q}\right)^2},
\label{eq:Ti}
\end{equation}
where $\bm{\gamma}_1, \ldots, \bm{\gamma}_L$ are $L$ independent draws, each obtained by first sampling the subset size $|\bm{\gamma}|$ uniformly from $\{2, \ldots, p-1\}$ and then selecting a subset of that size uniformly at random, and $\bar{q} = \frac{1}{L}\sum_{l=1}^L q(\bm{\gamma}_l)$. The factor $\frac{1}{4}$ in the numerator follows the normalization of \citet[Lemma~3]{becker2021varsel}, derived for $\gamma_i$ uniformly distributed on $\{0,1\}$. The population trigger index is defined analogously in terms of $q^*$:
\begin{equation}
\Ti = \frac{\mathbb{E}_{\bm{\gamma}}\left[(q^*(\bm{\gamma}^{(i)}) - q^*(\bm{\gamma}))^2\right]}{4\,\mathrm{Var}_{\bm{\gamma}}(q^*(\bm{\gamma}))}.
\label{eq:Ti_population}
\end{equation}
where $\mathbb{E}_{\bm{\gamma}}$ and $\mathrm{Var}_{\bm{\gamma}}$ are taken with respect to the sampling distribution of $\bm{\gamma}$ described above.

The plain-English interpretation of $\hatTi$ is: \emph{what fraction of the variability in predictive 
loss across input subsets is attributable to whether 
$X_i$ is included or excluded?} Let $d_l = q(\bm{\gamma}_l^{(i)}) - q(\bm{\gamma}_l)$ denote the marginal contribution of $X_i$ at the $l$-th draw. By the bias-variance decomposition of 
the second moment, $\frac{1}{L}\sum_{l=1}^{L} d_l^2 = \bar{d}^2 + 
\frac{L-1}{L}\widehat{\mathrm{Var}}(d)$, where $\bar{d} = \frac{1}{L}\sum_{l=1}^L d_l$ and $\widehat{\mathrm{Var}}(d) = \frac{1}{L-1}\sum_{l=1}^L(d_l - \bar{d})^2$. Therefore, $\hatTi$ captures both the average marginal contribution 
of $X_i$ to predictive loss and its contextual variability 
across different input subsets.

\paragraph{Consistency}	
Consistency of $\hatTi$ with respect to its population counterpart $\Ti$ requires both $L$ and $n_{tr}$ to grow: as $L \to \infty$, $\hatTi$ converges to a quantity that still depends on the fitted metamodel, which in turn converges to $\Ti$ as $n_{tr} \to \infty$. A second result establishes a variable-selection property under input
independence: $\Ti = 0$ whenever $X_i$ is irrelevant for $f$. We work under the following assumptions. 

\begin{assumption}[Finite second moment of $q$]
\label{ass:T1}
The predictive loss function $q(\bm{\gamma})$ satisfies $\mathbb{E}_{\bm{\gamma}}[q(\bm{\gamma})^2] < \infty$, where the expectation is taken with respect to the sampling distribution of $\bm{\gamma}$.
\end{assumption}

Assumption~\ref{ass:T1} ensures that the Law of Large Numbers applies to both the numerator and denominator of $\hatTi$.

\begin{assumption}[Positive denominators]
\label{ass:T2}
$\mathrm{Var}_{\bm{\gamma}}(q(\bm{\gamma})) > 0$ and $\mathrm{Var}_{\bm{\gamma}}(q^*(\bm{\gamma})) > 0$.
\end{assumption}

Assumption~\ref{ass:T2} excludes the degenerate case in which the predictive loss is constant across all subsets $\bm{\gamma}$; the second condition ensures that the population index $\Ti$ in equation~\eqref{eq:Ti_population} is well defined.

\begin{assumption}[Pointwise consistency of predictive loss]
\label{ass:T3}
For every $\bm{\gamma} \in \{0,1\}^p$ with $|\bm{\gamma}| \in \{2,\ldots,p-1\}$:
\begin{equation}
q(\bm{\gamma}) \xrightarrow{p} q^*(\bm{\gamma}) \quad \text{as } n_{tr} \to \infty.
\end{equation}
\end{assumption}

Assumption~\ref{ass:T3} requires that the predictive loss of the fitted metamodel converges to its population-optimal counterpart for every input subset $\bm{\gamma}$; since there are finitely many such subsets, this convergence holds jointly across all $\bm{\gamma}$, and in particular implies $q(\bm{\gamma}^{(i)}) - q(\bm{\gamma}) \xrightarrow{p} q^*(\bm{\gamma}^{(i)}) - q^*(\bm{\gamma})$ for every $\bm{\gamma}$.

\begin{theorem}
\label{thm:Ti_consistency}
Under Assumptions~\ref{ass:T1}--\ref{ass:T3}, $\hatTi \xrightarrow{p} \Ti$ as $n_{tr} \to \infty$ and $L \to \infty$.
\end{theorem}

\begin{proof}
Fix $n_{tr}$ and let $L \to \infty$. The numerator of $\hatTi$ is $\frac{1}{4L}\sum_{l=1}^L d_l^2$, where $d_l = q(\bm{\gamma}_l^{(i)}) - q(\bm{\gamma}_l)$ are i.i.d.\ since the $\bm{\gamma}_l$ are i.i.d.\ by construction. Under Assumption~\ref{ass:T1}, $\mathbb{E}[d_l^2] \leq 
2\mathbb{E}[q(\bm{\gamma}_l^{(i)})^2] + 
2\mathbb{E}[q(\bm{\gamma}_l)^2] < \infty$, 
where finiteness follows from Assumption~\ref{ass:T1} 
applied to both $\bm{\gamma}$ and $\bm{\gamma}^{(i)}$. The denominator $\frac{1}{L-1}\sum_{l=1}^L(q(\bm{\gamma}_l) - \bar{q})^2$ is the sample variance of $q(\bm{\gamma}_l)$, which converges in probability to $\mathrm{Var}_{\bm{\gamma}}(q(\bm{\gamma})) > 0$ under Assumptions~\ref{ass:T1} and~\ref{ass:T2} by the Law of Large Numbers. By the continuous mapping theorem applied to the ratio,
\begin{equation}
\hatTi \xrightarrow{p} T_{X_i}^{(s, n_{tr})} := \frac{\mathbb{E}_{\bm{\gamma}}\left[(q(\bm{\gamma}^{(i)}) - q(\bm{\gamma}))^2\right]}{4\,\mathrm{Var}_{\bm{\gamma}}(q(\bm{\gamma}))} \qquad \text{as } L \to \infty.
\end{equation}
This quantity $T_{X_i}^{(s, n_{tr})}$ is the trigger index computed as if the training sample of size $n_{tr}$ were fixed and infinitely many subsets were sampled; it depends on the fitted metamodel and hence, in general, on $n_{tr}$. Now we show that $T_{X_i}^{(s, n_{tr})} \xrightarrow{p} \Ti$ as $n_{tr} \to \infty$. Under Assumption~\ref{ass:T3}, $q(\bm{\gamma}) \xrightarrow{p} q^*(\bm{\gamma})$ for every $\bm{\gamma}$, jointly over the finite set of subsets $\{\bm{\gamma} : |\bm{\gamma}| \in \{2,\ldots,p-1\}\}$. By the continuous mapping theorem, the numerator and denominator of $T_{X_i}^{(s, n_{tr})}$ converge in probability to the numerator and denominator of $\Ti$ in equation~\eqref{eq:Ti_population}; since the latter denominator is positive by Assumption~\ref{ass:T2}, a further application of the continuous mapping theorem to the ratio gives $T_{X_i}^{(s, n_{tr})} \xrightarrow{p} \Ti$.
\end{proof}

\begin{theorem}
\label{thm:Ti_varsel}
Assume that the input variables are mutually independent. If $X_i$ is irrelevant for $f$, that is, $f(\mathbf{x}) = f(\mathbf{x}_{\sim i})$ for all $\mathbf{x} \in \mathcal{X}$, then $\Ti = 0$.
\end{theorem}

\begin{proof}
If $X_i$ is irrelevant for $f$ and the input variables are mutually independent, the inclusion of $X_i$ provides no additional information about the response for any predictor subset $\bm{\gamma}$. Therefore, the population-optimal predictive loss is unchanged when the inclusion indicator of $X_i$ is flipped:
\[
q^*(\bm{\gamma}^{(i)}) = q^*(\bm{\gamma})
\qquad \text{for every } \bm{\gamma}.
\]
Hence the numerator of $\Ti$ in equation~\eqref{eq:Ti_population} is zero,
and therefore $\Ti = 0$.
\end{proof}

\paragraph{Relationship with existing methods}

The proposed trigger index builds upon the GSA-based variable selection framework of \citet{becker2021varsel}, who introduced trigger sensitivity measures for identifying influential variables through random perturbations of model specifications. The present work extends this idea from deterministic model evaluations to observational settings, replacing repeated evaluations of the underlying model with predictive losses computed from supervised learning metamodels. As a result, the trigger framework becomes directly applicable when only observational data are available.

An interesting connection is that the trigger index bears a structural resemblance to the Shapley value \citep{shapley1953value,lundberg2017unified} computed on the predictive loss function, defined as:
\begin{equation}
\phi_i(q) = \sum_{S \subseteq \{1,\ldots,p\} \setminus \{i\}} \frac{|S|!(p-|S|-1)!}{p!}\left(q(S \cup \{i\}) - q(S)\right).
\label{eq:shapley}
\end{equation}
Both measures aggregate the marginal contribution of $X_i$ across input subsets, and drawing the subset size uniformly before drawing a subset of that size weights subsets in a way similar in spirit to the Shapley weights. The two measures differ in how marginal contributions are aggregated: the Shapley value averages linear differences, whereas $\Ti$ averages squared differences normalized by the variance of the predictive loss, thereby reflecting both the magnitude of the average marginal contribution of $X_i$ and its variation across predictor subsets. An empirical comparison between $\hatTi$ and Shapley-based importance is presented in Section~\ref{sec:mc_ti}, where, following common practice, Shapley importance is computed with TreeSHAP on the predictions of the fitted metamodel rather than on the predictive loss.

\section{Monte Carlo Studies}
\label{sec:simulations}

The Monte Carlo experiments are designed to complement the theoretical results established in the previous sections by investigating the finite-sample behavior of the proposed framework. For the first-order index $\Si$, we investigate whether the estimator (i) recovers the classical first-order Sobol' index when the assumptions of variance-based GSA hold, (ii) exhibits the convergence behavior predicted by the theoretical results, and (iii) remains informative when the input variables are no longer independent. For the trigger-based structural index $\hatTi$, we investigate whether it
(i) identifies structurally relevant inputs, (ii) distinguishes structural relevance from predictive association induced
by dependence, and (iii) characterizes its behavior under increasingly complex dependency structures.

Unless otherwise stated, all experiments are based on independent training and test samples of size $n_{tr}=n_{te}=1000$, and results are averaged over $R=500$ Monte Carlo replications. Three supervised learning algorithms are considered throughout: Random Forests (RF), Extreme Gradient Boosting (XGB), and feed-forward Neural Networks (NN). RF are fitted using the \texttt{ranger} package with 500 trees, \texttt{mtry}$=\lfloor\sqrt{p}\rfloor$, minimum node size 5, and bootstrap sampling with replacement. XGB models are fitted using the \texttt{xgboost} package with \texttt{nrounds}=100, \texttt{max\_depth}=6 and \texttt{eta}=0.1. NN are fitted using the \texttt{nnet} package with a single hidden layer of 10 units and \texttt{maxit}=500. The response is rescaled by its maximum absolute value prior to neural network training and back-transformed for prediction.

\subsection{First-Order Observational Index}
\label{sec:mc_si}

\subsubsection{Recovery of Classical Sobol' Indices}
\label{sec:bench}

Our first objective is to assess whether the proposed estimator recovers the classical first-order Sobol' indices when the assumptions underlying variance-based GSA are satisfied. We therefore consider benchmark functions for which the analytical values of the Sobol' indices are known exactly.

The Ishigami function \citep{ishigami1990importance} and the Sobol' G-function \citep{Saltelli1995} are used as the main benchmarks because they represent two complementary settings: the former combines strong nonlinearity with interaction effects, whereas the latter is a product-form function in which input importance decreases rapidly across inputs.
Table~\ref{tab:benchmark_si} reports the results. Additional benchmark results for the Bratley and Oakley--O'Hagan functions are reported in Appendix~\ref{app:bench}.

\begin{table}[htbp]
\centering
\caption{Recovery of the classical first-order Sobol' indices. Mean (standard deviation) of $\hat{S}_{X_i}$ over $R=500$ Monte Carlo replications. $S_i^*$ denotes the analytical first-order Sobol' index.}
\label{tab:benchmark_si}

\begin{tabular}{llcccc}
\toprule
 & Variable & NN & RF & XGB & $S_i^*$ \\
\midrule

\multirow{3}{*}{Ishigami}
& $X_1$ & 0.324 (0.022) & 0.432 (0.020) & 0.335 (0.016) & 0.314 \\
& $X_2$ & 0.464 (0.034) & 0.380 (0.022) & 0.426 (0.023) & 0.442 \\
& $X_3$ & 0.001 (0.002) & 0.010 (0.006) & 0.003 (0.001) & 0.000 \\
\midrule

\multirow{8}{*}{Sobol' G}
& $X_1$ & 0.747 (0.042) & 0.822 (0.028) & 0.782 (0.025) & 0.717 \\
& $X_2$ & 0.185 (0.013) & 0.152 (0.016) & 0.164 (0.010) & 0.179 \\
& $X_3$ & 0.023 (0.003) & 0.016 (0.005) & 0.009 (0.001) & 0.024 \\
& $X_4$ & 0.007 (0.002) & 0.006 (0.003) & 0.002 (0.000) & 0.007 \\
& $X_5$ & 0.000 (0.001) & 0.001 (0.001) & 0.000 (0.000) & 0.000 \\
& $X_6$ & 0.000 (0.001) & 0.001 (0.001) & 0.000 (0.000) & 0.000 \\
& $X_7$ & 0.000 (0.001) & 0.001 (0.001) & 0.000 (0.000) & 0.000 \\
& $X_8$ & 0.000 (0.002) & 0.001 (0.001) & 0.000 (0.000) & 0.000 \\
\bottomrule
\end{tabular}

\end{table}

Across both benchmark functions, $\hat{S}_{X_i}$ closely reproduces the analytical first-order Sobol' indices and correctly identifies the ranking of input importance across the three metamodel classes. Some differences in finite-sample accuracy remain across learners, with RF showing a moderate upward bias for the most influential inputs in some settings. Importantly, all three metamodels assign values close to zero whenever the analytical first-order index vanishes.

\subsubsection{Empirical Convergence}
\label{sec:cons}

Our second objective is to examine the empirical convergence behavior described by Theorem~\ref{thm:Si}. In particular, we investigate the two quantities appearing explicitly in the theoretical construction: the accuracy of the metamodel approximation, controlled by the training sample size $n_{tr}$, and the approximation of the main-effect function, controlled by the number of grid points $K$.

The study is conducted on the Ishigami and Sobol' G benchmark functions. Performance is measured by the mean absolute error (MAE) between the estimated and analytical first-order Sobol' indices over $R=100$ Monte Carlo replications.

We first fix the test sample size at $n_{te} = 1000$ and the grid size at $K = 100$, while varying the training sample size over $n_{tr} \in \{200, 500, 1000, 2000\}$. Then, we fix $n_{tr}=n_{te}=1000$ and investigate the effect of the grid resolution by varying $K\in\{50,100,150\}$.
The corresponding results are shown respectively in Figures~\ref{fig:cons_a} and~\ref{fig:cons_c}. For graphical clarity, both figures report the results for $X_1$ and $X_2$ only. In both benchmark functions these variables correspond to the largest non-zero first-order effects and therefore provide the most informative setting for assessing convergence. The remaining inputs have either substantially smaller or null analytical indices and exhibit the same qualitative behavior.

\begin{figure}[htbp]
\centering
\includegraphics[width=0.75\textwidth]{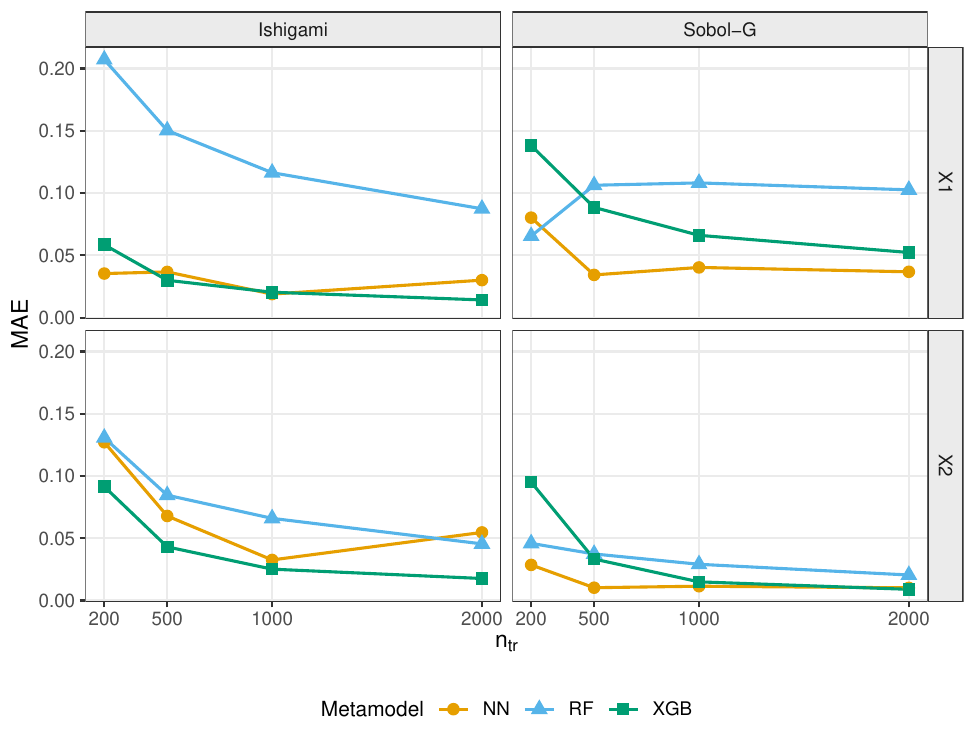}
\caption{
Empirical convergence of the observational first-order estimator as the training sample size increases. Mean absolute error (MAE) with respect to the analytical Sobol' indices is reported for the Ishigami and Sobol' G benchmark functions for variables $X_1$ and $X_2$ over $R=100$ Monte Carlo replications. The test sample size is fixed at $n_{te} = 1000$ and the quantile grid at $K = 100$.
}
\label{fig:cons_a}
\end{figure}

\begin{figure}[htbp]
\centering
\includegraphics[width=0.80\textwidth]{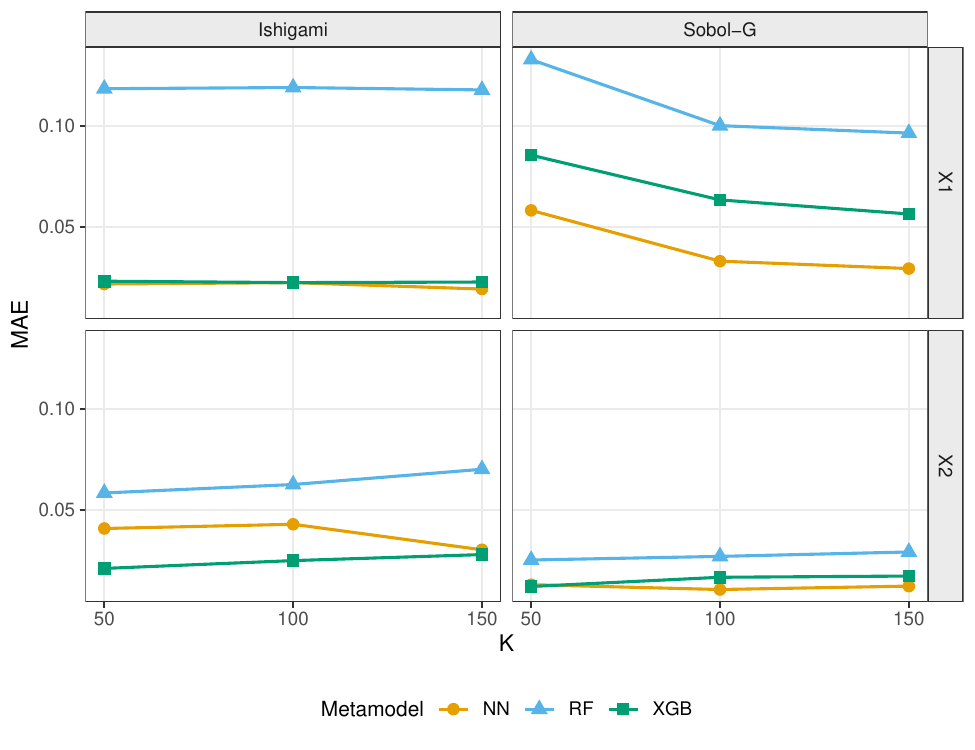}
\caption{
Effect of the quantile grid size on the observational first-order estimator. Mean absolute error (MAE) with respect to the analytical Sobol' indices is reported for the Ishigami and Sobol' G benchmark functions for variables $X_1$ and $X_2$ over $R=100$ Monte Carlo replications. Both training and test sample sizes are fixed at $n_{tr}=n_{te}=1000$.
}
\label{fig:cons_c}
\end{figure}

Increasing the training sample size reduces the estimation error for all three metamodels, in agreement with the role of metamodel approximation in Theorem~\ref{thm:Si}. The improvement is particularly pronounced for XGB and NN, whereas RF exhibits a moderate residual bias, in agreement with the benchmark results of the previous section.

The effect of the grid size is considerably smaller. Moving from $K=50$ to $K=100$ produces a noticeable reduction in MAE, whereas little additional improvement is observed for $K=150$. This behavior suggests that a grid of approximately one hundred quantile points provides an adequate approximation of the main-effect function for the sample sizes considered here.

\subsubsection{Dependency Structures}
\label{sec:dag}

We now investigate observational settings in which the input variables are dependent and the classical variance decomposition is no longer directly applicable. 
We note that Theorem~\ref{thm:Si} establishes consistency of $\hat{S}_{X_i}$ toward $S_i^f$ under input independence. In the scenarios considered below, the inputs are dependent by construction. Consequently, $\hat{S}_{X_i}$ no longer represents a component of the classical Sobol' variance decomposition; it instead estimates the variance of the marginal average-effect function, normalized by the variance of the systematic response.

Two representative directed acyclic graph (DAG) structures are considered. In both cases, some variables are associated with the response only indirectly through intermediate predictors, whereas others enter the response directly. These scenarios allow us to examine whether, under the specific dependency structures considered here, the estimator separates predictors entering the response directly from predictors associated with it through intermediate variables.

Throughout this study we set $a=0.7$ and $b=3$, with all error terms independently distributed as $\epsilon_j\sim\mathcal N(0,1)$.

\textsc{Scenario 1}
\begin{equation}
\begin{split}
X_1, X_2 &= \epsilon_1,\epsilon_2,\\
X_3 &= aX_1+\epsilon_3,\\
X_4 &= aX_2+\epsilon_4,\\
Y &= bX_3+bX_4+\epsilon_Y.
\end{split}\label{scenario1}
\end{equation}

The response depends directly on $X_3$ and $X_4$. Variables $X_1$ and $X_2$ are associated with the response only indirectly through the intermediate variables $X_3$ and $X_4$, respectively. Consequently, the proposed estimator should assign high values to $X_3$ and $X_4$, while assigning negligible values to $X_1$ and $X_2$.

\textsc{Scenario 2}
\begin{equation}
\begin{split}
X_1,\ldots,X_6 &= \epsilon_1,\ldots,\epsilon_6,\\
X_7 &= a(X_1+X_2+X_3+X_4+X_5+X_6)+\epsilon_7,\\
Y &= bX_7+\epsilon_Y.
\end{split}\label{scenario2}
\end{equation}

The response depends directly only on $X_7$. Variables $X_1,\ldots,X_6$ are associated with the response exclusively through their contribution to the common intermediate variable $X_7$. This scenario therefore evaluates whether the proposed estimator can distinguish a direct association with the response from an indirect association induced by the dependency structure.

Table~\ref{tab:dag} reports the estimated observational first-order indices.

\begin{table}[htbp]
\centering
\caption{First-order indices $\hat{S}_{X_i}$ under representative dependency structures. Entries report the mean (standard deviation) over $R=500$ Monte Carlo replications.}
\label{tab:dag}

\begin{tabular}{llccc}
\toprule
 & Variable & NN & RF & XGB\\
\midrule

\multirow{4}{*}{\textsc{Scenario 1}}
& $X_1$ & 0.003 (0.018) & 0.002 (0.001) & 0.000 (0.000)\\
& $X_2$ & 0.003 (0.028) & 0.002 (0.001) & 0.000 (0.000)\\
& $X_3$ & 0.559 (0.055) & 0.467 (0.024) & 0.514 (0.028)\\
& $X_4$ & 0.564 (0.194) & 0.467 (0.024) & 0.514 (0.027)\\
\midrule

\multirow{7}{*}{\textsc{Scenario 2}}
& $X_1$ & 0.000 (0.001) & 0.005 (0.001) & 0.000 (0.000)\\
& $X_2$ & 0.001 (0.008) & 0.005 (0.001) & 0.000 (0.000)\\
& $X_3$ & 0.001 (0.004) & 0.005 (0.001) & 0.000 (0.000)\\
& $X_4$ & 0.001 (0.009) & 0.005 (0.001) & 0.000 (0.000)\\
& $X_5$ & 0.001 (0.007) & 0.005 (0.001) & 0.000 (0.000)\\
& $X_6$ & 0.001 (0.012) & 0.005 (0.001) & 0.000 (0.000)\\
& $X_7$ & 1.090 (0.249) & 0.651 (0.017) & 1.047 (0.024)\\
\bottomrule
\end{tabular}

\end{table}

The results clearly distinguish variables directly associated with the response from those whose association is entirely mediated by the dependency structure. In Scenario 1, all three metamodels assign negligible observational first-order indices to $X_1$ and $X_2$, while correctly identifying $X_3$ and $X_4$ as the only variables directly contributing to the variability of the response. A similar pattern is observed in Scenario 2, where the six background variables receive values close to zero despite being strongly associated with $Y$ through the common mediator $X_7$, whereas $X_7$ receives values close to one. Values slightly above one may occur in these dependent-input settings because the resulting quantity is no longer a component of a classical variance decomposition and is therefore not constrained to the unit interval.

These experiments illustrate the behavior of the proposed estimator when the independence assumption underlying the classical Sobol' decomposition is relaxed. In the dependency structures considered here, the observational first-order index provides a clear separation between variables that contribute directly to response variability and variables whose association with the response is mediated by other predictors. These results should not be interpreted as recovering a classical Sobol' decomposition, since under input dependence $\hat{S}_{X_i}$ targets the marginal average-effect quantity described above rather than $S_i$.

As an additional robustness assessment, Appendix~\ref{app:becker} reports the results of the Becker stress-test protocol, in which the proposed estimator is evaluated over a large ensemble of randomly generated metafunctions spanning a wide range of nonlinear and interaction structures.

\subsection{Trigger-Based Structural Index}
\label{sec:mc_ti}
The second part of the Monte Carlo study investigates the empirical behavior of the structural index $\hatTi$. Here, the objective is to characterize the type of variable relevance captured by the trigger construction. Because the trigger requires repeated refitting of the metamodel across randomly sampled predictor subsets, this part of the study is restricted to RF and XGB, with lighter configurations to contain the computational cost: RF with 100 trees and XGB with \texttt{nrounds}=50, \texttt{max\_depth}=4 and \texttt{eta}=0.1. Each replication uses a sample of size $n=1000$ and $L=100$ random inclusion vectors, and results are averaged over $R=500$ Monte Carlo replications. The predictive loss $q(\bm{\gamma})$ is computed as an out-of-sample RMSE: for RF we use out-of-bag predictions on the full sample, whereas for XGB the model is trained on one half of the sample and the loss is evaluated on the other half. For comparison, we also report the global TreeSHAP importance \citep{lundberg2020}, defined as the mean absolute SHAP value across observations, computed on RF and XGB fits with the same configurations.

\subsubsection{Structural Relevance in Benchmark Functions}
\label{sec:ti_bench}

We first consider the Ishigami and Sobol' G-functions in order to assess whether the $\hatTi$ index identifies inputs that are structurally relevant to the response and assigns negligible values to variables that do not enter the data-generating mechanism.

The Ishigami function provides a useful stress case because $X_3$ contributes only through an interaction term and has zero first-order Sobol' index. The Sobol' G-function, in contrast, contains several inputs with exactly zero contribution, allowing direct assessment of the variable-selection property established in Theorem~\ref{thm:Ti_varsel}.

Table~\ref{tab:ti_benchmark} reports the mean and standard deviation of $\hatTi$ and TreeSHAP importance over the Monte Carlo replications.

\begin{table}[htbp]
\centering
\caption{Trigger-based structural sensitivity and TreeSHAP importance for the Ishigami and Sobol' G benchmark functions. Entries report mean (standard deviation) over $R=500$ Monte Carlo replications.}
\label{tab:ti_benchmark}

\begin{tabular}{llcccc}
\toprule
 & Variable & $\hatTi$ RF & $\hatTi$ XGB & SHAP RF & SHAP XGB \\
\midrule
\multirow{3}{*}{Ishigami}
& $X_1$ & 0.898 (0.131) & 1.125 (0.220) & 1.755 (0.078) & 1.777 (0.087) \\
& $X_2$ & 0.858 (0.116) & 0.970 (0.198) & 1.637 (0.053) & 1.859 (0.053) \\
& $X_3$ & 0.205 (0.038) & 0.203 (0.050) & 0.526 (0.029) & 0.542 (0.044) \\
\midrule
\multirow{8}{*}{Sobol' G}
& $X_1$ & 0.922 (0.055) & 0.878 (0.049) & 0.366 (0.013) & 0.477 (0.014) \\
& $X_2$ & 0.092 (0.015) & 0.131 (0.024) & 0.159 (0.010) & 0.212 (0.009) \\
& $X_3$ & 0.005 (0.001) & 0.002 (0.001) & 0.048 (0.008) & 0.045 (0.006) \\
& $X_4$ & 0.005 (0.001) & 0.000 (0.000) & 0.029 (0.006) & 0.014 (0.004) \\
& $X_5$ & 0.005 (0.001) & 0.000 (0.000) & 0.017 (0.003) & 0.003 (0.002) \\
& $X_6$ & 0.005 (0.001) & 0.000 (0.000) & 0.017 (0.003) & 0.003 (0.002) \\
& $X_7$ & 0.005 (0.001) & 0.000 (0.000) & 0.017 (0.003) & 0.003 (0.002) \\
& $X_8$ & 0.005 (0.001) & 0.000 (0.000) & 0.017 (0.003) & 0.003 (0.002) \\
\bottomrule
\end{tabular}
\end{table}

The results show that $\hatTi$ sharply separates structurally relevant from irrelevant inputs. In the Sobol' G-function, variables $X_5,\ldots,X_8$, which do not enter the data-generating mechanism, are assigned values essentially equal to zero, particularly under XGB. This behavior is in line with the variable-selection result established in Theorem~\ref{thm:Ti_varsel}. The Ishigami function illustrates a characteristic feature of the trigger index. Variable $X_3$ contributes to the response exclusively through interaction effects and therefore has a smaller structural contribution than $X_1$ and $X_2$. Nevertheless, both RF and XGB assign a non-zero trigger value to $X_3$, reflecting its active role in the data-generating mechanism. The trigger index therefore distinguishes between variables that are structurally involved in the response and variables that are completely irrelevant, irrespective of the specific form through which they contribute. Overall, both the trigger index and TreeSHAP correctly identify the relevant variables in these benchmark settings.

\subsubsection{Structural Relevance in the Presence of Correlated Predictors}
\label{sec:ti_correlato}

We next investigate whether $\hatTi$ can distinguish a variable that enters the response directly from one that is predictive only because it is correlated with a relevant input. We consider the following scenario: \\
\textsc{Scenario 3}
\begin{equation}
\begin{split}
(X_1,X_2) &\sim \mathcal N_2(\mathbf 0,\Sigma_\rho),\\
X_3 &= 0.7X_1+\epsilon_3,\\
X_4 &= \epsilon_4,\\
Y &= 3X_1+3X_2+\epsilon_Y,
\end{split}
\end{equation}
where $X_1$ and $X_2$ enter the response directly, $X_3$ is a proxy associated with $Y$ only through its correlation with $X_1$, and $X_4$ is independent noise. The experiment is repeated for $\rho\in\{0,0.3,0.6,0.9\}$.

\begin{table}[htbp]
\centering
\caption{Scenario 3: Trigger index $\hatTi$ and TreeSHAP importance under increasing correlation between the two directly associated predictors. Entries report mean (standard deviation) over $R=500$ Monte Carlo replications.}
\label{tab:scenario3}

\begin{tabular}{llcccc}
\toprule
$\rho$ & Variable & $\hatTi$ RF & $\hatTi$ XGB & SHAP RF & SHAP XGB\\
\midrule

\multirow{4}{*}{0.0}
& $X_1$ & 0.667 (0.078) & 0.585 (0.077) & 2.161 (0.078) & 2.333 (0.091) \\
& $X_2$ & 0.906 (0.088) & 0.811 (0.082) & 2.313 (0.073) & 2.345 (0.095) \\
& $X_3$ & 0.044 (0.011) & 0.020 (0.007) & 0.310 (0.047) & 0.059 (0.024) \\
& $X_4$ & 0.007 (0.002) & 0.000 (0.000) & 0.069 (0.006) & 0.046 (0.014) \\
\midrule

\multirow{4}{*}{0.3}
& $X_1$ & 0.743 (0.102) & 0.635 (0.091) & 2.347 (0.099) & 2.512 (0.174) \\
& $X_2$ & 0.957 (0.117) & 0.840 (0.101) & 2.452 (0.092) & 2.489 (0.161) \\
& $X_3$ & 0.043 (0.011) & 0.018 (0.007) & 0.319 (0.046) & 0.056 (0.018) \\
& $X_4$ & 0.006 (0.002) & 0.000 (0.000) & 0.067 (0.005) & 0.048 (0.014) \\
\midrule

\multirow{4}{*}{0.6}
& $X_1$ & 0.916 (0.178) & 0.788 (0.150) & 2.414 (0.123) & 2.555 (0.257) \\
& $X_2$ & 1.069 (0.209) & 0.932 (0.178) & 2.503 (0.112) & 2.571 (0.260) \\
& $X_3$ & 0.039 (0.010) & 0.015 (0.006) & 0.323 (0.051) & 0.060 (0.017) \\
& $X_4$ & 0.005 (0.002) & 0.000 (0.000) & 0.065 (0.006) & 0.052 (0.016) \\
\midrule

\multirow{4}{*}{0.9}
& $X_1$ & 1.416 (0.542) & 1.283 (0.484) & 2.395 (0.192) & 2.502 (0.440) \\
& $X_2$ & 1.445 (0.497) & 1.308 (0.451) & 2.395 (0.190) & 2.439 (0.433) \\
& $X_3$ & 0.044 (0.011) & 0.017 (0.006) & 0.300 (0.056) & 0.067 (0.020) \\
& $X_4$ & 0.005 (0.002) & 0.000 (0.000) & 0.065 (0.005) & 0.058 (0.016) \\
\bottomrule
\end{tabular}

\vspace{1mm}
\footnotesize
\textit{Note.} $X_1$ and $X_2$ enter the data-generating mechanism directly; $X_3$ is associated with the response only through its correlation with $X_1$; $X_4$ is independent noise.
\end{table}

Table~\ref{tab:scenario3} summarizes the results. Across all correlation levels, $\hatTi$ preserves a clear separation between the directly relevant inputs and the proxy variable. The mean trigger value assigned to $X_3$ remains small and stable as $\rho$ increases, while $X_1$ and $X_2$ retain substantially larger values. The irrelevant variable $X_4$ is driven essentially to zero. TreeSHAP instead assigns a positive importance to the proxy variable, which is attributable to its predictive interpretation.

These results show that the trigger index is sensitive to structural relevance rather than association alone: a predictor may be statistically informative about the response without receiving a large $\hatTi$ if its predictive content is redundant with variables already present in the model.

\subsubsection{Dependency Structures}
\label{sec:ti_dag}

Finally, we examine the behavior of $\hatTi$ when a variable that does not enter the response carries information about several directly relevant predictors. We consider two scenarios that share the same dependency structure and differ only in the functional form of the response.

\textsc{Scenario 4}
\begin{equation}
\begin{split}
X_1 &= \epsilon_1,\\
X_j &= aX_1+\epsilon_j,\qquad j=2,3,4,\\
Y &= bX_2+bX_3+bX_4+\epsilon_Y.
\end{split}
\label{scenario4}
\end{equation}

In this scenario, the response depends directly on $X_2$, $X_3$, and $X_4$, whereas $X_1$ is associated with the response only through these three variables. Because $X_1$ is linearly related to all three directly associated predictors, it can partially reconstruct the predictive signal whenever one or more of them are omitted from a sampled predictor subset.

\textsc{Scenario 4$'$}
\begin{equation}
\begin{split}
X_1 &= \epsilon_1,\\
X_j &= aX_1+\epsilon_j,\qquad j=2,3,4,\\
Y &= bX_2^2+bX_3+bX_4+\epsilon_Y.
\end{split}
\label{scenario4b}
\end{equation}

Scenario~4$'$ preserves the same dependency structure as Scenario~4, but introduces a nonlinear, quadratic effect of $X_2$ on $Y$, while $X_3$ and $X_4$ retain a linear effect and the generating equations for the mediators are unchanged.

Table~\ref{tab:ti_dag} reports the corresponding values of $\hatTi$ together with TreeSHAP importance.

\begin{table}[htbp]
\centering
\caption{Trigger-based structural sensitivity $\hatTi$ and TreeSHAP importance under Scenarios~4 and~4$'$. Entries report mean (standard deviation) over $R=500$ Monte Carlo replications.}
\label{tab:ti_dag}

\begin{tabular}{llcccc}
\toprule
 & Variable & $\hatTi$ RF & $\hatTi$ XGB & SHAP RF & SHAP XGB\\
\midrule

\multirow{4}{*}{\textsc{Scenario 4}}
& $X_1$ & 0.265 (0.074) & 0.140 (0.050) & 1.601 (0.279) & 1.069 (0.400) \\
& $X_2$ & 1.025 (0.184) & 0.870 (0.166) & 2.662 (0.191) & 2.788 (0.267) \\
& $X_3$ & 1.046 (0.192) & 0.887 (0.170) & 2.652 (0.201) & 2.791 (0.260) \\
& $X_4$ & 1.061 (0.196) & 0.898 (0.173) & 2.675 (0.199) & 2.792 (0.253) \\
\midrule

\multirow{4}{*}{\textsc{Scenario 4$'$}}
& $X_1$ & 0.050 (0.016) & 0.017 (0.010) & 0.545 (0.175) & 0.177 (0.113) \\
& $X_2$ & 1.470 (0.141) & 1.262 (0.105) & 3.755 (0.230) & 4.099 (0.285) \\
& $X_3$ & 0.201 (0.050) & 0.149 (0.050) & 2.657 (0.179) & 2.598 (0.172) \\
& $X_4$ & 0.203 (0.052) & 0.148 (0.049) & 2.659 (0.160) & 2.607 (0.160) \\
\bottomrule
\end{tabular}
\vspace{1mm}
\footnotesize
\noindent
\begin{minipage}{\textwidth}
\textit{Note.} In both scenarios $X_2$, $X_3$ and $X_4$ are directly associated with the response, whereas $X_1$ is associated with it only through these three variables. In Scenario~4, $X_1$ carries linear information about all three directly associated variables. Scenario~4$'$ replaces the linear effect of $X_2$ on $Y$ with a quadratic one, weakening this redundancy.
\end{minipage}
\end{table}

Under the linear specification of Scenario~4, both $\hatTi$ and TreeSHAP assign a non-negligible importance to $X_1$. Although $X_1$ is only indirectly associated with the response, it acts as an effective predictive proxy for the omitted mediators and therefore retains predictive value across competing models.

The quadratic specification of Scenario~4$'$ weakens this redundancy while preserving the same dependency structure. As a consequence, the trigger value assigned to $X_1$ collapses towards zero, whereas the directly associated variables remain clearly distinguishable. TreeSHAP also reduces the importance assigned to $X_1$, although the separation between the proxy variable and the directly associated predictors is less pronounced.

Taken together, these experiments clarify the interpretation of the proposed structural index. $\hatTi$ quantifies how strongly predictive performance depends on the availability of a predictor across alternative predictor subsets. Its value therefore reflects both the predictor's own contribution and the extent to which its information can substitute for, or be substituted by, information carried by other predictors.

\section{Illustrative Case Study}
\label{sec:application}

We illustrate MM--GSA using data from the National Health and Nutrition Examination Survey (NHANES), a nationally representative health survey conducted by the U.S. National Center for Health Statistics. We use the \texttt{NHANES} R package \citep{nhanes2015}, which contains 10{,}000 observations resampled from the 2009--2010 and 2011--2012 survey cycles to account for the oversampling of specific population groups. This case study is performed to show how $\Si$ and $\hatTi$ behave on real, correlated covariates. No epidemiological conclusions about systolic blood pressure are drawn from this case study.

The outcome is average systolic blood pressure (\texttt{BPSysAve}). To keep the illustration interpretable, we retain eight continuous predictors covering demographic, anthropometric, physiological, metabolic, and socioeconomic characteristics: age, body mass index (BMI), height, pulse rate, total cholesterol, direct cholesterol, average sleep duration, and the poverty-to-income ratio. Starting from a deduplicated sample of $n=6779$ individuals, observations with missing values in any of the retained variables were excluded, yielding a final analytic sample of $n=4178$ individuals.

For each of $R = 100$ repeated 70/30 train--test splits, we fit RF, XGB, and NN using the same specifications as in Section~\ref{sec:simulations}, except that for the neural network predictors are standardized using training-set means and standard deviations and a weight decay of $0.01$ is applied, to stabilize training on covariates with heterogeneous scales. We compute $\Si$, $\hatTi$ (with $L = 100$ random inclusion vectors, restricted to RF and XGB as in Section~\ref{sec:mc_ti}), and TreeSHAP importance on each split. Table~\ref{tab:nhanes_results} reports the mean and standard deviation of each sensitivity measure across replications.

\begin{table}[htbp]
\centering
\caption{NHANES application: mean (standard deviation) of $\Si$, $\hatTi$, and TreeSHAP importance across $R=100$ repeated 70/30 train--test splits.}
\label{tab:nhanes_results}

\begin{tabular}{lcccc}
\toprule
Variable & $\Si$ RF & $\Si$ XGB & $\Si$ NN & $\hatTi$ RF \\
\midrule
Age            & 0.637 (0.031)  & 0.403 (0.029)  & 0.566 (0.036)  & 0.893 (0.051)  \\
BMI            & 0.044 (0.008)  & 0.054 (0.013)  & 0.060 (0.055)  & 0.022 (0.002)  \\
Height         & 0.035 (0.007)  & 0.037 (0.009)  & 0.037 (0.009)  & 0.027 (0.003)  \\
TotChol        & 0.018 (0.005)  & 0.022 (0.006)  & 0.021 (0.018)  & 0.013 (0.002)  \\
DirectChol     & 0.017 (0.005)  & 0.033 (0.016)  & 0.087 (0.068)  & 0.011 (0.002)  \\
SleepHrsNight  & 0.024 (0.013)  & 0.027 (0.015)  & 0.066 (0.066)  & 0.017 (0.004)  \\
Poverty        & 0.010 (0.003)  & 0.024 (0.006)  & 0.014 (0.005)  & 0.015 (0.002)  \\
Pulse          & 0.015 (0.007)  & 0.026 (0.024)  & 0.013 (0.013)  & 0.014 (0.003)  \\
\bottomrule
\end{tabular}

\vspace{2mm}

\begin{tabular}{lccc}
\toprule
Variable & $\hatTi$ XGB & SHAP RF & SHAP XGB \\
\midrule
Age            & 0.949 (0.050)  & 5.95 (0.06)  & 6.23 (0.16) \\
BMI            & 0.011 (0.006)  & 1.76 (0.04)  & 1.57 (0.13) \\
Height         & 0.012 (0.006)  & 1.72 (0.03)  & 1.66 (0.15) \\
TotChol        & 0.005 (0.003)  & 1.30 (0.04)  & 0.76 (0.13) \\
DirectChol     & 0.002 (0.001)  & 0.81 (0.02)  & 0.48 (0.08) \\
SleepHrsNight  & 0.001 (0.001)  & 0.37 (0.01)  & 0.28 (0.07) \\
Poverty        & 0.004 (0.002)  & 1.05 (0.03)  & 0.93 (0.15) \\
Pulse          & 0.002 (0.001)  & 0.82 (0.03)  & 0.41 (0.10) \\
\bottomrule
\end{tabular}
\end{table}

Age ranks first under $\Si$, $\hatTi$, and TreeSHAP for both RF and XGB. For the remaining predictors, however, the rankings differ across measures.

This distinction is particularly visible under RF. Total cholesterol ranks fourth according to TreeSHAP but seventh according to $\hatTi$, whereas sleep duration shows the opposite pattern, ranking last according to TreeSHAP and fourth according to $\hatTi$. Thus, a predictor may contribute substantially to model predictions while its availability across alternative predictor subsets has comparatively little impact on predictive performance, and vice versa. Under XGB the two rankings are more similar. This illustrates the different information conveyed by the two measures: TreeSHAP describes feature attribution within the fitted model, whereas $\hatTi$ measures the sensitivity of predictive loss to the inclusion or exclusion of the predictor.

More generally, the rankings produced by the different measures are related but not interchangeable. Across the repeated splits, the mean Spearman correlation between $\Si$ and $\hatTi$ is $0.699$ under RF and $0.549$ under XGB. The application therefore illustrates the motivation for considering the two MM--GSA quantities jointly: $\Si$ quantifies the contribution of $X_i$ to the variability of the fitted systematic response, whereas $\hatTi$ quantifies how strongly predictive performance depends on the availability of $X_i$ across alternative predictor subsets.

\section{Conclusion}
\label{sec:conclusion}

This paper has proposed MM--GSA, a metamodel-based approach to Global Sensitivity Analysis from given data. Supervised learning metamodels provide access to global sensitivity questions when the underlying input--output mechanism cannot be repeatedly evaluated under a designed sampling scheme. MM--GSA combines two distinct components. The first provides a model-agnostic estimator $\Si$ of the first-order index $S_i^f$, which coincides with the classical first-order Sobol' index under input independence and a noiseless response, and quantifies the contribution of $X_i$ to the variability of the systematic response. The second introduces the new structural sensitivity index $\Ti$, together with its estimator $\hatTi$, which quantifies how strongly predictive performance depends on the availability of $X_i$ across alternative predictor subsets. Consistency results are established for both estimators, together with a variable-selection property for $\Ti$ under the conditions of Theorem~\ref{thm:Ti_varsel}.

Monte Carlo experiments across benchmark functions and dependent input structures illustrate the finite-sample behavior of the two measures and their different interpretations. Comparisons with TreeSHAP further show that predictive attribution and sensitivity to predictor availability need not convey the same information. The NHANES case study shows that these differences also arise with real correlated covariates. In particular, predictors receiving non-negligible attribution within a fitted model may have a comparatively limited impact on predictive loss when their availability is varied across alternative predictor subsets. The application therefore provides an empirical illustration of why $\Si$ and $\hatTi$ can be informative when considered jointly: $\Si$ characterizes the contribution of an input to systematic response variability, while $\hatTi$ characterizes the dependence of predictive performance on the availability of that input.

MM--GSA is closely related to the given-data perspective of \citet{borgonovo2016common}, discussed in relation to $\Si$ in Section~\ref{sec:s_i_estimator}. Their work establishes a general route to estimating probabilistic sensitivity measures from a single input--output sample through partition-refining strategies. MM--GSA addresses the given-data setting through metamodel-based constructions and, through $\Ti$, introduces predictor availability as an additional dimension of sensitivity.

Several directions remain open for future research. The step-function construction underlying $\Si$ is one of several possible nonparametric approximations to the main-effect function $f_i$. Smoother alternatives, such as kernel smoothing or local regression, could be considered in its place. We leave the study of such alternatives, and their implications for the consistency and finite-sample behavior of $\Si$, to future work.

Theorem~\ref{thm:Ti_varsel} establishes a sufficient condition for a vanishing structural index under input independence. A complete characterization of the converse direction remains open. More generally, extending the theoretical characterization of $\Ti$ beyond input independence requires accounting explicitly for dependence and redundancy among the available predictors.

The dependency scenarios considered in Section~\ref{sec:ti_dag} show that the behavior of $\hatTi$ is affected by the extent to which information carried by one predictor can substitute for information carried by others. Scenario~4 and its nonlinear variant illustrate that this phenomenon can depend on the functional form of the response relationship and on the ability of the metamodel to exploit the available predictive information. A systematic characterization of structural sensitivity under different forms of dependence and predictive redundancy therefore represents an important direction for further research.

Finally, the presence--absence criterion underlying $\hatTi$ could be extended to other probabilistic sensitivity measures covered by the common rationale of \citet{borgonovo2016common}. In the present work, this criterion is applied to predictive loss, so that $\hatTi$ quantifies the sensitivity of predictive performance to the availability of $X_i$ across alternative predictor subsets. A natural extension would retain the same presence--absence construction while considering, for example, density-based, distribution-based, or information-based sensitivity measures. This would make it possible to investigate whether the structural relevance of predictor availability changes with the particular feature of the output distribution targeted by the sensitivity measure.

\section*{Acknowledgments}
The author thanks A. S. and R. S. for their helpful discussions on the theoretical results and their guidance throughout this work. 

\section*{Funding}
This research was funded by the Italian Ministry of Research, in complementary actions to the NRRP Fit4MedRob - Fit for Medical Robotics Grant (\# PNC0000007).

\section*{Data Availability Statement}
The NHANES data used in the empirical application are publicly available through the \texttt{NHANES} R package \citep{nhanes2015}. Code implementing the proposed estimators, together with the scripts used for the Monte Carlo experiments and the empirical application, is available from the author upon request and will be made publicly available upon publication.

\section*{Conflict of Interest Statement}
The author declares no conflict of interest.

\section*{Artificial Intelligence Disclosure}
During the preparation of this work, the author used AI tools to assist with editing and refining the manuscript text and checking internal notational and grammar style. The author reviewed and takes full responsibility for the content of this publication.

\bibliographystyle{plainnat}
\bibliography{bibfile}

\appendix

\section{Additional Benchmark Functions}
\label{app:bench}

This appendix reports additional benchmark results for the observational first-order estimator $\hat{S}_{X_i}$. The Bratley function and the Oakley--O'Hagan function complement the benchmark scenarios presented in Section~\ref{sec:mc_si} by considering higher-dimensional models with interactions and mixed linear--nonlinear effects. All three metamodels recover the ordering of the inputs. Neural Networks provide the closest agreement with the reference values, while Random Forests and, to a lesser extent, XGBoost overestimate the indices of the most influential inputs, in line with the results in Section~\ref{sec:bench}.
\paragraph{Bratley Function}

Table~\ref{tab:bratley} reports the mean and standard deviation of $\hat{S}_{X_i}$ over $R=500$ Monte Carlo replications for the Bratley function \citep{bratley1992implementation}. The function is an alternating sum of products of the inputs, with rapidly decreasing first-order contributions, providing a useful complement to the Ishigami and Sobol' G benchmarks discussed in the main text.

\begin{table}[htbp]
	\centering
	\caption{Recovery of the first-order Sobol' indices for the Bratley function. Mean (standard deviation) of $\hat{S}_{X_i}$ over $R=500$ Monte Carlo replications. $S_i^*$ denotes the reference first-order Sobol' index, computed numerically.}
	\label{tab:bratley}
	\begin{tabular}{lcccc}
		\toprule
		Variable & NN & RF & XGB & $S_i^*$ \\
		\midrule
		$X_1$ & 0.684 (0.024) & 0.748 (0.025) & 0.703 (0.024) & 0.659 \\
		$X_2$ & 0.175 (0.010) & 0.164 (0.011) & 0.171 (0.010) & 0.168 \\
		$X_3$ & 0.042 (0.003) & 0.032 (0.005) & 0.036 (0.003) & 0.040 \\
		$X_4$ & 0.012 (0.001) & 0.007 (0.002) & 0.008 (0.001) & 0.011 \\
		$X_5$ & 0.002 (0.000) & 0.001 (0.001) & 0.001 (0.000) & 0.002 \\
		$X_6$ & 0.001 (0.000) & 0.001 (0.000) & 0.000 (0.000) & 0.001 \\
		$X_7$ & 0.000 (0.000) & 0.000 (0.000) & 0.000 (0.000) & 0.000 \\
		$X_8$ & 0.000 (0.000) & 0.000 (0.000) & 0.000 (0.000) & 0.000 \\
		\bottomrule
	\end{tabular}
\end{table}

\paragraph{Oakley--O'Hagan Function}

The Oakley--O'Hagan function \citep{oakley2004probabilistic} represents a higher-dimensional benchmark ($p=15$) combining linear, nonlinear and interaction effects. The first ten variables have relatively small first-order Sobol' indices, whereas the last five account for most of the output variance. Table~\ref{tab:oakley} reports the corresponding estimates of $\hat{S}_{X_i}$.

\begin{table}[htbp]
	\centering
	\caption{Recovery of the first-order Sobol' indices for the Oakley--O'Hagan function. Mean (standard deviation) of $\hat{S}_{X_i}$ over $R=500$ Monte Carlo replications. $S_i^*$ denotes the reference first-order Sobol' index, computed numerically.}
	\label{tab:oakley}
	\begin{tabular}{lcccc}
		\toprule
		Variable & NN & RF & XGB & $S_i^*$ \\
		\midrule
		$X_1$  & 0.002 (0.001) & 0.002 (0.002) & 0.002 (0.002) & 0.002 \\
		$X_2$  & 0.001 (0.001) & 0.001 (0.001) & 0.001 (0.001) & 0.000 \\
		$X_3$  & 0.002 (0.001) & 0.001 (0.001) & 0.002 (0.001) & 0.001 \\
		$X_4$  & 0.003 (0.001) & 0.002 (0.002) & 0.002 (0.001) & 0.002 \\
		$X_5$  & 0.002 (0.001) & 0.002 (0.001) & 0.001 (0.001) & 0.004 \\
		$X_6$  & 0.021 (0.003) & 0.018 (0.007) & 0.018 (0.005) & 0.023 \\
		$X_7$  & 0.025 (0.004) & 0.017 (0.007) & 0.018 (0.005) & 0.024 \\
		$X_8$  & 0.029 (0.005) & 0.023 (0.008) & 0.026 (0.007) & 0.026 \\
		$X_9$  & 0.043 (0.006) & 0.035 (0.012) & 0.035 (0.010) & 0.046 \\
		$X_{10}$ & 0.016 (0.003) & 0.010 (0.005) & 0.010 (0.004) & 0.014 \\
		\midrule
		$X_{11}$ & 0.106 (0.011) & 0.147 (0.029) & 0.132 (0.022) & 0.102 \\
		$X_{12}$ & 0.141 (0.012) & 0.221 (0.035) & 0.195 (0.026) & 0.134 \\
		$X_{13}$ & 0.106 (0.010) & 0.148 (0.030) & 0.130 (0.021) & 0.103 \\
		$X_{14}$ & 0.112 (0.011) & 0.152 (0.031) & 0.138 (0.023) & 0.105 \\
		$X_{15}$ & 0.132 (0.012) & 0.192 (0.033) & 0.169 (0.025) & 0.120 \\
		\bottomrule
	\end{tabular}
\end{table}

\section{Stress Test on Randomly Generated Metafunctions}
\label{app:becker}

To complement the benchmark experiments reported in the main text, we evaluate the robustness of the proposed first-order estimator using the stress-test protocol of \citet{becker_metafunctions_2020}, subsequently adopted by \citet{puy_comprehensive_2022}. Unlike the benchmark functions considered in Section~\ref{sec:mc_si}, this protocol generates a large ensemble of randomly constructed metafunctions spanning a wide range of nonlinearities and interaction structures.

For each Monte Carlo replication, a new metafunction is generated and reference first-order Sobol' indices are estimated using the estimator of \citet{saltelli2010a} applied to a large Monte Carlo sample of $N=2000$ observations drawn from the same metafunction. These estimates are taken as reference values against which the proposed observational estimator $\hat{S}_{X_i}$ is evaluated. The estimator is then computed from an independent sample of $n_{tr}+n_{te}=2000$ observations using a quantile grid of size $K=100$.

Performance is summarized by the mean absolute error (MAE) between the estimated and reference indices together with the Spearman rank correlation $\rho$, both averaged over $R=500$ Monte Carlo replications with $p=10$ input variables.

Table~\ref{tab:becker} reports the results. All three metamodels achieve low estimation error and high rank correlation across the randomly generated metafunctions, indicating that the proposed estimator remains accurate beyond the classical benchmark functions considered in the main text. XGB provides the lowest MAE and the highest rank correlation, followed by NN and RF. The small variability of the MAE across replications further indicates stable performance over a broad range of function structures.

\begin{table}[htbp]
\centering
\caption{Performance of $\hat{S}_{X_i}$ under the Becker stress-test protocol. Results are averaged over $R=500$ randomly generated metafunctions with $p=10$ input variables.}
\label{tab:becker}

\begin{tabular}{lccc}
\toprule
Metamodel & MAE & Spearman $\rho$ & SD(MAE)\\
\midrule
XGB & 0.012 & 0.960 & 0.005\\
NN  & 0.014 & 0.918 & 0.027\\
RF  & 0.019 & 0.937 & 0.008\\
\bottomrule
\end{tabular}
\end{table}

\end{document}